\documentclass[conference]{IEEEtran}
\usepackage[utf8]{inputenc}
\usepackage{amsmath,amsthm,amssymb,graphicx,booktabs,balance}
\usepackage{enumerate}
\usepackage{enumitem}
\usepackage{microtype}
\usepackage{color}
\usepackage{xcolor}
\usepackage{url}
\newtheorem{theorem}{Theorem}[section]
\newtheorem{lemma}[theorem]{Lemma}

\newtheorem{definition}[theorem]{Definition}

\newtheorem{remark}[theorem]{Remark}

\begin{document}

\title{Game-Theoretic Analysis of Transaction Selection in DAG-Based Distributed Ledgers}

\author{
    \IEEEauthorblockN{Sebastian Müller\IEEEauthorrefmark{1}\IEEEauthorrefmark{2}, Alexandre Reiffers-Masson\IEEEauthorrefmark{3}}
    \IEEEauthorblockA{\IEEEauthorrefmark{1}I2M, UMR 7373, Aix-Marseille Université, CNRS, Centrale Marseille, Marseille, 13453, France \\
    Email: sebastian.muller@univ-amu.fr}
    \IEEEauthorblockA{\IEEEauthorrefmark{2}IOTA Foundation, Berlin, 10405, Germany}
    \IEEEauthorblockA{\IEEEauthorrefmark{3}IMT Atlantique, LabSTICC, UMR CNRS 6285, Brest, F-29238, France \\
    Email: alexandre.reiffers-masson@imt-atlantique.fr}
}

\date{}

\maketitle

\begin{abstract}
Transaction selection in parallel or DAG-based distributed ledger technologies (DLTs) is a crucial challenge that directly impacts throughput, fairness, and validator incentives. In these systems, validators independently choose transactions to include in their blocks, often relying on naive heuristics like uniform or proportional selection. This can lead to inefficient outcomes when validators prioritize their own rewards without considering collective impacts.

We analyze two fee allocation mechanisms used in practice: Random Fee Allocation (RFA), where transaction fees are randomly assigned to one validator, and Collaborative Fee Sharing (CFS), where fees are distributed equally among all validators. Using a single-shot game-theoretic framework, we derive symmetric Nash equilibria (NE) for selecting transactions for both mechanisms and propose an optimization-based method to compute these equilibria. Numerical simulations demonstrate that the NE of CFS consistently achieves higher throughput and rewards compared to the NE of RFA, particularly under skewed fee distributions.

Additionally, we compare these equilibrium strategies to naive benchmarks (uniform and proportional selection), showing that the proportional strategy outperforms the NE of RSA in many situations. These findings may provide actionable insights into the design of transaction selection and incentive mechanisms, enabling more robust and high-performance DAG-based DLTs.
\end{abstract}

\section{Introduction}

DAG-based distributed ledger technologies (DLTs) have emerged as a promising alternative to traditional blockchain architectures. A key advantage is their ability to let multiple validators propose blocks concurrently, improving throughput. However, this parallel block creation introduces new challenges, such as overlapping (or colliding) transactions.

Collisions occur when multiple validators select the same transaction, reducing system efficiency since it is unnecessary to include the same transaction more than once. 

In DAG-based blockchain, transaction selection strategies are chosen based on simplicity, such as random or proportional transaction selection.  Such choices are influenced by the type of consensus used, for instance, permissionless versus committee-based consensus, see \cite{SoKDAG}. 
Moreover, the fee distribution mechanism is typically predefined without considering how the system's actors may react to it, leaving the assumption that these intuitive strategies are close to Nash equilibria in terms of performance largely unexamined. 

In this paper, we adopt a conceptual approach to bridge this gap by starting from the fee distribution mechanisms and deriving the \textit{mixed symmetric Nash equilibrium} (NE) for transaction selection strategies. Symmetric NEs are particularly relevant in decentralized and permissionless systems, as they do not require communication or coordination among validators. We can see such equilibria as a sampling mechanism, suggesting to all the actors to select their transactions, with the guarantee that they cannot do better if everyone uses it.  

Our work focuses on a \emph{single-shot} transaction selection model, also referred to as myopic in prior works (e.g., \cite{inclusive}), where each validator selects a bounded number of transactions from a finite shared pool. This model abstracts away the temporal dynamics of real-world transaction arrival and prioritizes understanding the equilibrium behaviour of validators when making a single, isolated decision. Importantly, we make the following assumptions:
\begin{enumerate}
    \item All validators share identical transaction buffers or mempools.
    \item The block’s position within the block DAG does not depend on the creator’s identity.\footnote{This assumption is particularly relevant for the Collaborative Fee Sharing (CFS) mechanism but can become problematic in Random Fee Allocation (RFA), where block ordering may influence the reward distribution.}
\end{enumerate}

Our analysis is focusing on two natural fee allocation mechanisms:
\begin{itemize}
    \item \emph{Random Fee Allocation (RFA)}: The fee of a transaction included by multiple validators is randomly assigned to one of those that have included the transaction in their blocks.
    \item \emph{Collaborative Fee Sharing (CFS)}: The fee is distributed equally among all validators regardless of whether the validator included the transaction.
\end{itemize}

For both models, we derive the mixed symmetric Nash equilibrium and evaluate their performance in terms of fee throughput (the total fees collected) and transaction throughput (the number of unique transactions included). Additionally, we compare these strategies with naive benchmarks, such as uniform and proportional transaction selection, to assess their practicality.

\subsection*{Key Contributions:}
\begin{itemize}
    \item We derive symmetric Nash equilibria (NE) for both Collaborative Fee Sharing (CFS) and Random Fee Allocation (RFA) mechanisms in a single-shot transaction selection model.
    \item Through numerical evaluations, we show that the NE of {CFS consistently outperforms other strategies}.
    \item In certain scenarios, proportional transaction selection achieves rewards comparable to the Nash equilibrium in RFA. This is particularly relevant for protocol designs where CFS is not feasible, such as GHOSTDAG \cite{Sompolinsky2018PHANTOMG} or other permissionless PoW-based systems, making proportional selection a valuable heuristic.
\end{itemize}

\section{Fee Allocation Models}
\label{sec:fee_models}

Fee allocation mechanisms determine how validators are rewarded when multiple blocks include the same transaction, influencing system performance. We study the two following models:

\subsection{Random Fee Allocation (RFA)}
Under RFA, if $N$ validators include transaction $T$, one is randomly chosen to receive the full fee, while the others get nothing. This mechanism creates a trade-off: including high-fee transactions increases potential rewards but also raises the risk of collisions, reducing expected payoffs. This fee allocation is typical to PoW-systems where the collusion is decided by a total ordering, for instance, in GHOSTDAG, e.g.~\cite{Sompolinsky2018PHANTOMG}.

\subsection{ Collaborative Fee Sharing (CFS)}
Under CFS, if at least one validator includes T, its fee is shared equally among all N validators. Including T ensures a validator contributes to triggering the reward but has little additional benefit if the transaction is already likely to be included by others. This reduces the need to include high-fee transactions as it is sufficient that one of the $N$ validators includes and potentially increases the total throughput of the system. This fee allocation is more typical for BFT-like protocols that require knowledge of the set of validators.

\section{Single-Shot Inclusive-F Game Model}
\label{sec:inclusive_game}

We consider a set of m transactions $\{w_1, w_2, \dots, w_m\},$ each associated with a fee $v(w_i) > 0$. These fees are drawn from a discrete set of values $\{v_1, v_2, \dots, v_n\}$, where $v_1 > v_2 > \cdots > v_n.$\footnote{In blockchain systems, transaction fees are typically integers (e.g., satoshis or wei), ensuring practical enforceability.} Validators select a subset of $b$ transactions to include in their blocks.

A validator’s strategy is modelled as a \emph{mixed strategy}, represented either as:
\begin{enumerate}
	\item A probability distribution over all subsets of size $b$, or equivalently,
	\item A vector of marginal probabilities $\{q_i\}$, where $0 \leq q_i \leq 1$ and $\sum_{i=1}^m q_i = b$, capturing the likelihood of including transaction $w_i$ in the block.
\end{enumerate}
This equivalence, shown in \cite{inclusive}, is central to our analysis and enables a more tractable formulation of the problem.

\begin{definition}[Induced Marginal Distribution,]
Let $p$ be a probability over all $b$-subsets of an $m$-element set. The 
\emph{induced marginal distribution} $q=\{q_i\}\in [0,1]^m$ is given by 
\[
  q_i = \sum_{\substack{S\subseteq\{1,\dots,m\}\\|S|=b,\;i\in S}} p_S.
\]
Note that $q\in [0,1]^m$, so $\sum_{i=1}^m q_i=b$.  
\end{definition}

Conversely, any $q$ with coordinate sum~$b$ can be realized by an appropriate 
distribution over subsets:
\begin{lemma} (see Lemma 9 in extended version of \cite{inclusive})
Any vector $q\in[0,1]^m$ with $\sum_{i=1}^m q_i=b$ arises as the marginal 
distribution $q(p)$ of some $p$ in $\Delta(\binom{m}{b})$.
\end{lemma}

\subsection{Mixed Symmetric Nash Equilibria (NE)}
We focus on mixed symmetric NE, where validators use identical strategies. These equilibria are well-suited for distributed systems as they avoid the need for communication or trusted setups. 
For each model, we derive two descriptions of the symmetric NE; see Section \ref{sec:theorem6} and \ref{sec:efficient_NE}. The first description is due to \cite{inclusive} and provides an interesting theoretical result that allows us to understand the nature of the NE. As calculating the actual NE using Theorem \ref{thm:NE_general} seems unfeasible, we propose a description of the NEs using a concave optimization problem in Section \ref{sec:efficient_NE}.

\section{Symmetric NE for Single-Shot Inclusive-F Game}
\label{sec:theorem6}
We suppose that transactions can be grouped by 
\emph{fee levels} $v_1\ge \dots\ge v_n$, with $k_\ell$ transactions 
paying $v_\ell$. Note that  $m=\sum_{\ell=1}^n k_\ell$.  We assume that the transactions are ordered in decreasing order by fee, i.e.,  $v(w_1)\ge v(w_2)\ge\dots\ge v(w_m)$.

We assume there's a function 
$f:[0,1]\to\mathbb{R}$ that is \emph{strictly decreasing} 
and invertible, representing a player's \emph{expected share of the fee rewards} of transaction $w_i$ if it is included with marginal probability $p_i$ 
by all validators.  

\begin{theorem}[Symmetric NE]
\label{thm:NE_general}
Fix $n$ fee levels $\{v_1,\dots,v_n\}$ with $k_\ell$ transactions paying fee 
$v_\ell$. Let $m=\sum_{\ell=1}^n k_\ell$, and sort these $m$ transactions by descending 
fees $w_1\ge\dots\ge w_m$. Let $\ell(w_i)$ be the index $\ell$ s.t. $v(w_i)=v_\ell$.

Let $f$ be strictly decreasing on $[0,1]$, with an inverse $f^{-1}$. We define 
\[G_\ell(z)=\sum_{h=1}^\ell k_h\,\min(f^{-1}(\tfrac{z}{v_h}),1)-b.\]
Then 
\[
   p_i := \frac{q_{\ell(w_i)}}{k_{\ell(w_i)}}\quad(i=1,\dots,m),
\]
defines a symmetric Nash equilibrium in the single-shot inclusive-F game
where $q_\ell$ is chosen according to partial coverage conditions, i.e.,
\[
  q_\ell = 
   \begin{cases}
     k_\ell\min\Bigl(f^{-1}(\tfrac{c_{k_{\max}}}{v_\ell}),1\Bigr) & \text{if }1\le \ell\le k_{\max},\\[6pt]
     0 & \text{if }k_{\max}< \ell\le n,
   \end{cases}
\]
with 

\[k_{\max} := \max\{\,k\le n:\,G_\ell(v_\ell)\le0\ \forall\,l\le k\},\] 
and $c_{k_{\max}}$ is the real root of $G_{k_{\max}}$. 
\end{theorem}
The previous result states that no player can 
increase its payoff by transferring probability mass from $p_i$ in $(0,1)$ 
to some $p_j<1$, since $\{p_i\}$ is a symmetric Nash equilibrium.

The idea of the proof is to show $(0<p_i<1)\implies 
v(w_i)\,f(p_i)\ge v(w_j)\,f(p_j)$, splitting by whether $\ell(i)\le k_{\max}$ 
and $\ell(j)>k_{\max}$. The condition $\sum_i p_i=b$ is ensured by 
$c_{k_{\max}}$ being the root of $G_{k_{\max}}$.

\begin{remark}
In the case of {homogeneous fees}, where all \(v_\ell\) are equal, the strategy simplifies significantly. All transactions receive equal probability weights, as no transaction is distinguished by its fee level. When \(k_{\max}\) includes all transactions, the equilibrium probabilities depend on the shape of the function \(f\), which dictates how probability is distributed among transactions to meet the marginal condition \(\sum p_i = b\).
However, low-fee transactions are excluded entirely when \(k_{\max}\) is smaller. This ensures block capacity is allocated efficiently to higher-fee transactions in the Nash equilibrium. The function \(f\) and the fee structure jointly determine the allocation strategy.
\end{remark}

Theorem~\ref{thm:NE_general} provides a general framework for analyzing Nash equilibria in single-shot inclusive fee games. In this section, we apply the theorem to our two specific models of transaction selection: RFA and CFS. 

\subsection{Model 1: Random Selection (RFA)}
\label{sec:RFAmodel}
In RFA, if $m$ validators pick transaction $w$, each obtains $w/m$ in 
expectation. The function $f(p)$ captures a single validator's expected 
share from including $w$ (when others do include it with marginal probability $p$)  is
\begin{align}\label{eq:shareRFA}
    f(p) &= \sum_{i=0}^{N-1} \frac{1}{i+1} {N-1 \choose i}  p_i^i (1-p_i)^{N-1-i} \cr
    & = \frac{1 - (1-p)^{N}}{N\, p}.
\end{align}
is \emph{strictly 
decreasing} in $p$. For the last equality we used Lemma \ref{lem:identity}.\footnote{Note also that $f(p)=\mathbb{E}[\frac1{X+1}]$ with $X$ following a binomial distribution with parameters $p$ and $N-1$.}
Therefore, Theorem~\ref{thm:NE_general} applies. 
The equilibrium sets $p_i$ for each transaction $w_i$ according to 
$\min( f^{-1} (c_{k_{\max}}/v_{\ell(i)}),1)$, leading to partial coverage of top items 
and $p_i=0$ on lower items.

\begin{remark}
The expected reward remains the same in an alternative scenario where the fee $v(w)$ is equally split among all validators who include $w$. In this case, the expected share of reward is still given by $f(p)$, and the Nash equilibrium is identical to the one derived for RFA.
\end{remark}

\begin{remark}
Computing the exact equilibrium for RFA is computationally intensive due to the need to invert $f(p)$ and determine the coverage threshold $k_{\max}$. This complexity motivates the approach developed in Section \ref{sec:efficient_NE}.
\end{remark}

\subsection{Model 2: Collaborative Fee Sharing (CFS)}
\label{sec:CFSmodel}
Here, if at least one validator picks $w$, all get $w/N$. Then, from a single validator's 
standpoint, embedding $w$ is beneficial only if it might otherwise be 
omitted. The chance 
that none of the other validators picks $w$ is $(1-p)^{N-1}$. Hence, the expected share of the reward is 
\begin{equation}\label{eq:fCFS}
    f(p) = \frac1N (1-p)^{N-1}
\end{equation}
which is strictly 
decreasing in $p$.
In this case, we can explicitly calculate the inverse function of $f$:
\begin{equation}\label{eq:sfinverseCFS}
    f^{-1}(x) = 1- (Nx)^{\frac1{N-1}}.
\end{equation}
 Collisions do not matter as strongly, so 
$f$ typically is \emph{less} punishing than in RFA. But the monotonic decrease remains: if $p$ is large, coverage is nearly guaranteed, so embedding $w$ yourself confers a little extra payoff.  

Theorem~\ref{thm:NE_general} also yields a partial coverage threshold 
$c_{k_{\max}}$ and $p_i$ values. The difference is that the $f$ in CFS is 
less punishing for collisions, so $k_{\max}$ might be smaller than 
in RFA. As we will see in Section \ref{sec:evaluation}, more items might get partially covered in RSA, raising effective throughput and lowering fee throughput.

\section{Efficient Derivation of Nash Equilibria for RFA and CFS}
\label{sec:efficient_NE}

In this section, we present an efficient method to derive Nash equilibria (NE) for the Random Fee Allocation (RFA) and Collaborative Fee Sharing (CFS) models using the Karush-Kuhn-Tucker (KKT) conditions. We first recall the KKT conditions and their role in constrained optimization. Then, we apply them to derive the symmetric Nash equilibria for both RFA and CFS models.

\subsection{Recap: Karush-Kuhn-Tucker (KKT) Conditions}
The KKT conditions provide necessary and sufficient criteria for optimality in constrained optimization problems, particularly when the objective function is convex (or concave) and the feasible set is convex. Consider the optimization problem:
$$
\max_x \phi(x)$$ under the contraints
$$ \quad g_i(x) \leq 0, \; i=1,\dots,m, \quad h_j(x) = 0, \; j=1,\dots,p,
$$
where \(\phi(x)\) is the objective function, \(g_i(x)\) are inequality constraints, and \(h_j(x)\) are equality constraints.

The KKT conditions state that for a feasible point \(x^*\) to be optimal, there must exist Lagrange multipliers \(\mu_i \geq 0\) (for inequality constraints) and \(\lambda_j\) (for equality constraints) such that:
\begin{enumerate}
    \item \emph{Stationarity:} $$\nabla \phi(x^*) + \sum_{i=1}^m \mu_i \nabla g_i(x^*) + \sum_{j=1}^p \lambda_j \nabla h_j(x^*) = 0.$$
    \item \emph{Primal feasibility:} \(g_i(x^*) \leq 0, \; h_j(x^*) = 0\).
    \item \emph{Dual feasibility:} \(\mu_i \leq 0, \; \forall i\).
    \item \emph{Complementary slackness:} \(\mu_i g_i(x^*) = 0, \; \forall i\).
\end{enumerate}
If, in addition, the objective function $\phi$ is concave/convex, then the feasible point $x^*$ is a global optimum. 
\subsection{Application to Random Fee Allocation (RFA)}

In the RFA model, the expected reward for including a transaction \(i\) is proportional to the probability of being the unique validator to include it. For a symmetric setting, where all validators use the same strategy, this probability depends on the marginal inclusion probability \(p_i\). The optimization problem for a single validator \(k\) is:
\[
\max_{p_{i,k}} \sum_{i=1}^m p_{i,k} v(w_i) \alpha_i, \text{~s.t.~} 0\leq p_{i,k} \leq 1, \; \forall i, \sum_{i=1}^m p_{i,k} = b,
\]
where \(\alpha_i\) is the collision adjustment factor:
\[
\alpha_i = \frac{1 - (1-p_i)^{N}}{N p_i}.
\]

The KKT conditions can be written as:
\begin{align}
v(w_i) \alpha_i  - \mu_{i,k}^{0}+\mu_{i,k}^{1} + \lambda  = 0, \; \forall i, &\mbox{ (stationarity)} \label{eq:kkt_rfa1} \\
0 \leq p_{i,k} \leq 1,  \; \forall i, \sum_{i=1}^m p_{i,k}= b, &\mbox{ (primal feasibility)}\label{eq:kkt_rfa4} \\
\mu_{i,k}^0, \mu_{i,k}^1  \leq 0, \; \forall i,  &\mbox{ (dual feasibility)}\label{eq:kkt_rfa30} \\
\mu_{i,k}^0 p_i = \mu_{i,k}^1 (1-p_{i,k})=  0, \; \forall i. &\mbox{ (slackness)}\label{eq:kkt_rfa20} 
\end{align}
\subsubsection{Symmetric Equilibrium for RFA}
In symmetric equilibrium, we assume \(p_i^k = p_i\) for all validators \(k\). Hence, the equilibrium is obtained by solving:
\begin{align}
v(w_i) \alpha_i  - \mu_{i,k}^{0}+\mu_{i,k}^{1} + \lambda  &= 0, \; \forall i, \label{eq:kkt_rfa1sym} \\
\mu_i^0 p_i &= 0, \; \forall i, \label{eq:kkt_rfa20sym} \\
\mu_i^1 (1-p_i) &= 0, \; \forall i, \label{eq:kkt_rfa21sym} \\
\mu_i^0 &\leq 0, \; \forall i, \label{eq:kkt_rfa30sym} \\
\mu_i^1 &\leq 0, \; \forall i, \label{eq:kkt_rfa31sym} \\
0 \leq p_i &\leq 1, \; \forall i, \label{eq:kkt_rfa4sym} \\
\sum_{i=1}^m p_i &= b. \label{eq:kkt_rfa5sym}
\end{align}
Note that $\alpha_i$ is decreasing in $p_i$ implying that the conditions \eqref{eq:kkt_rfa1sym}-\eqref{eq:kkt_rfa5sym} are also the optimality conditions of the following concave optimization problem:
\[
\max_{p_i} \sum_{i=1}^m v(w_i) \int_0^{p_i}\hat \alpha(q_i) dq_i, \text{ s.t. } 0 \leq p_i \leq 1, \; \forall i, \sum_{i=1}^m p_i = b,
\]
with
\begin{equation*}
\hat \alpha(q_i):=\left\{\begin{array}{lll}
 \frac{1 - (1-q_i)^{N}}{N q_i}&\text{if}& q_i\in(0,1],\\
 0  &\text{if} & q_i = 0.
\end{array}\right.
\end{equation*}
Now, we can compute the mixed equilibria using a standard algorithm, such as projected gradient descent. 
\subsection{Application to Collaborative Fee Sharing (CFS)}

In the CFS model, the expected reward for a transaction \(i\) depends on the probability that no validator has included it. More precisely, for a given transaction $w_i$, validator $k$ obtains $v(w_i)/N$ of the reward if at least one of the validator did include it. 

The optimization problem for a single validator \(k\) becomes:
\[
\max_{p_{i,k}} \sum_{i=1}^m \frac{v(w_i)}{N} \left( 1- \left((1- p_{i,k})  \prod_{j \neq k} (1 - p_{i,j})\right)\right)\]
under the constraints 
\[  0 \leq p_{i,k} \leq 1, \; \forall i, \quad \sum_{i=1}^m p_{i,k} = b. 
\]
As before, we obtain the symmetric Nash equilibria as the solution to the following concave optimisation problem:
$$
\max_{p_i}  \sum_{i=1}^m   \frac{v(w_i)}{N} (1-  (1 - p_i)^{N}),$$
under the constraints \[  0 \leq p_{i} \leq 1, \; \forall i, \quad \sum_{i=1}^m p_{i} = b.\]

\section{Numerical Evaluation}
\label{sec:evaluation}

In this section, we numerically compare the equilibria derived for the 
\emph{Random Fee Allocation (RFA)} model and the \emph{Collaborative Fee Sharing (CFS)} 
model under a Zipf-like fee distribution. Specifically, we sample 
transaction fees according to
\[
   \mu_i \;=\; \frac{C}{i^s}
   \quad\text{for } i=1,\ldots, maxFee,
\]
where \(s>0\) is a shape parameter and \(C = \sum_{i=1}^{maxFee} i^\alpha \) is the normalizing constant. We then compute numerically the 
symmetric Nash equilibrium (NE) probabilities for both RFA and CFS. We consider the following selection methods as benchmark strategies :
\begin{itemize}
    \item \emph{Random Transaction Selection (RTS)}: transactions are chosen uniformly, i.e. $q_i= b/m$.
    \item \emph{Proportional Transaction Selection (PTS)}: transaction are chosen proportionally to their fee, i.e.,
    \[
    q_i = b \frac{v(w_i)}{\sum_{j=1}^m v(w_j)},
    \] assuming here that $q_i \in [0,1].$
\end{itemize}

We measure the two following ``throughputs'':
\begin{enumerate}[label=(\roman*)]
    \item \emph{Effective transaction throughput}: the expected number of unique transactions includes in at least one block:
    \(\Theta_{tx} = 
        \sum_{i=1}^m [\,1 - (1 - q_i)^N\,].\)
    \item \emph{Effective fee throughput}: the expected rewards collected by all blocks 
    \(\Theta_{fee} = 
        \sum_{i=1}^m v(w_i) [\,1 - (1 - q_i)^N\,].\)
\end{enumerate}

\subsection{Default Parameters}
The following parameters are used as defaults in all simulations unless otherwise specified:
\begin{itemize}
    \item Number of validators: \(N=10\).
    \item Block capacity: \(b=100\) (each validator selects exactly \(b\) transactions per block).
    \item Total number of transactions: \(m=1000\).
    \item Maximum transaction fee: \(maxFee=10\).
    \item Fee distribution: uniform over \(\{1, \ldots, maxFee\}\), corresponding to a Zipf parameter \(s=0\).
    \item Number of simulations: \(sim=50\).\footnote{Due to the high number of transactions, fluctuations between simulation runs are minimal, justifying this choice.}
\end{itemize}

\subsection{Efficient Computation of Nash Equilibria}
To compute the Nash equilibria for both models, we employ projected gradient descent using \texttt{numpy}. For the CFS strategy, an exact solution can be obtained via convex optimization libraries such as \texttt{cvxpy}. However, due to the high computational complexity of exact methods, we opt for an efficient approximation. 

\subsubsection{Collaborative Fee Sharing (CFS)}
In the CFS model, the Nash equilibrium is obtained by solving the concave optimization problem:
\[
\max_{p}  - \sum_{i=1}^m v(w_i) (1 - p_i)^N,
\]
subject to:
\[
\sum_{i=1}^m p_i = b, \quad 0 \leq p_i \leq 1 \; \forall i.
\]

\subsubsection{Random Fee Allocation (RFA)}
In the RFA model, the computation of Nash equilibria involves solving the following optimization problem:
\[
\max_{p} \; \sum_{i=1}^m v(w_i) \int_{0}^{p_i}\frac{1 - (1 - q_i)^N}{N q_i}dq_i,
\]
subject to:
\[
\sum_{i=1}^m p_i = b, \quad 0 \leq p_i \leq 1 \; \forall i.
\]

\subsection{Validation and Convergence}
Convergence is validated by monitoring the stability of probabilities $p_i$ over iterations. Both the RFA and CFS models use projected gradient descent, with stopping criteria based on stabilization within a predefined tolerance. This iterative approach is preferred over exact methods for CFS due to computational constraints.

\subsection{Impact of Varying the Number of Transactions}
\label{sec:results_m_variation}

To analyze the performance of different strategies as the number of transactions (\(m\)) increases, we take the default parameters and vary the total number of transactions from \(100\) to \(10\,000\). 

Figures~\ref{fig:fee_throughput_m} and~\ref{fig:effective_throughput_m} illustrate the results for the expected fee throughput and effective throughput, respectively.

\paragraph{Expected Fee Throughput} 
As shown in Figure~\ref{fig:fee_throughput_m}, the Collaborative Fee Sharing (CFS) model achieves the highest fee throughput across all values of \(m\). The Proportional Transaction Selected (PTS)  model also performs well but consistently outperforms Random Fee Allocation (RFA). The Random Transaction Selction (RTS) also slightly outperforms RFA. 

\paragraph{Effective Throughput}
The results for effective throughput, shown in Figure~\ref{fig:effective_throughput_m}, reveal a distinct trade-off. The uniform strategy (RTS) achieves the highest effective throughput, maximizing the number of unique transactions included due to its even distribution of marginal probabilities. However, this comes at the cost of lower fee throughput. PTS and RFA yield similar effective throughputs, while CFS, despite achieving the highest fee throughput, has the lowest effective throughput. This can be attributed to the relatively low penalty for transaction collisions.

\begin{figure}[t]
    \centering
    \includegraphics[width=0.8\linewidth]{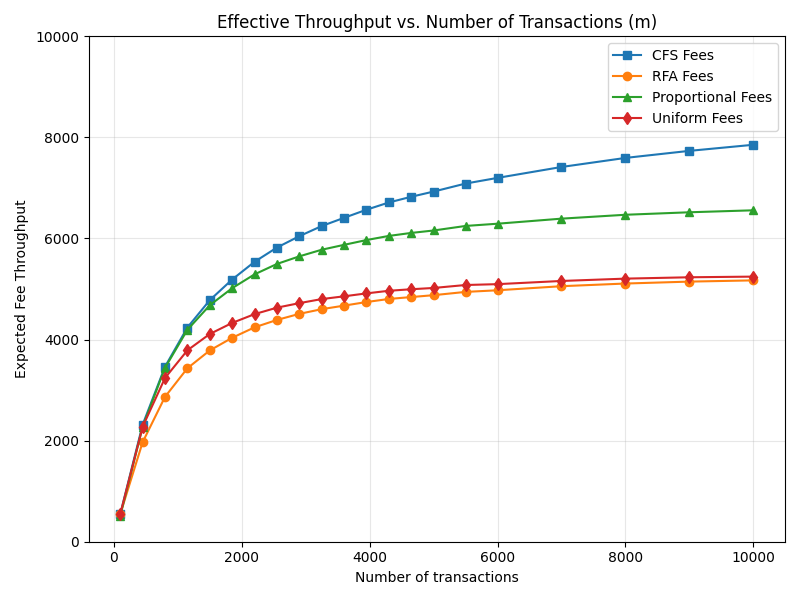}
    \caption{Expected Fee Throughput vs. Number of Transactions (\(m\)) for different strategies.}
    \label{fig:fee_throughput_m}
\end{figure}

\begin{figure}[t]
    \centering
    \includegraphics[width=0.8\linewidth]{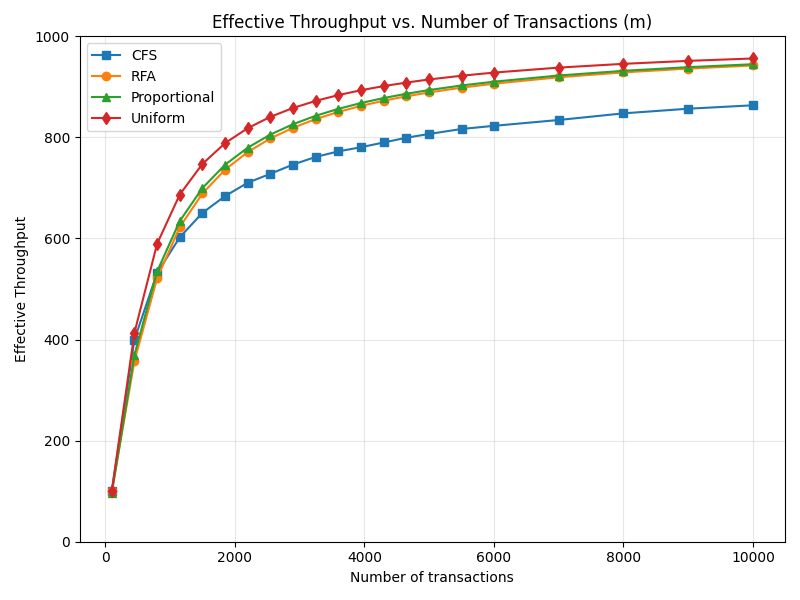}
    \caption{Effective Throughput vs. Number of Transactions (\(m\)) for different strategies.}
    \label{fig:effective_throughput_m}
\end{figure}

\subsection{Impact of Varying the Maximum Fee (\(\text{maxFee}\))}
\label{sec:results_maxFee_variation}

To analyze the performance of different strategies as the maximum fee (\(\text{maxFee}\)) increases, we fix the default parameters and vary \(\text{maxFee}\) from \(5\) to \(100\) in increments of \(5\).

Figures~\ref{fig:fee_throughput_maxFee} and~\ref{fig:effective_throughput_maxFee} illustrate the expected fee throughput and effective throughput results, respectively.

\paragraph{Expected Fee Throughput}
As shown in Figure~\ref{fig:fee_throughput_maxFee}, the CFS and PTS models achieve the highest fee throughput across all values of \text{maxFee}, effectively prioritizing high-fee transactions. In contrast, the Random Fee Allocation (RFA) model performs the worst, with its throughput declining as \text{maxFee} increases, highlighting its inefficiency in handling diverse fee distributions.

\paragraph{Effective Throughput}
Figure~\ref{fig:effective_throughput_maxFee} shows that the uniform strategy (RTS) consistently achieves the highest effective throughput by evenly distributing marginal probabilities across all transactions. However, this comes at the expense of lower fee throughput. The other three strategies exhibit similar performance, with a clear convergence, indicating that increasing \text{maxFee} has little asymptotic effect on effective throughput.

\begin{figure}[htbp]
    \centering
    \includegraphics[width=0.8\linewidth]{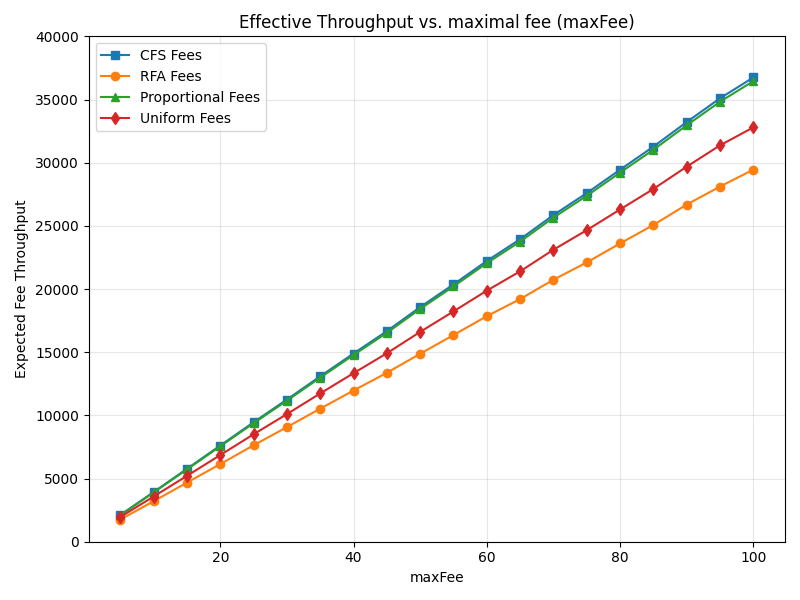}
    \caption{Expected Fee Throughput vs. Maximum Fee (\(\text{maxFee}\)).}
    \label{fig:fee_throughput_maxFee}
\end{figure}

\begin{figure}[htbp]
    \centering
    \includegraphics[width=0.8\linewidth]{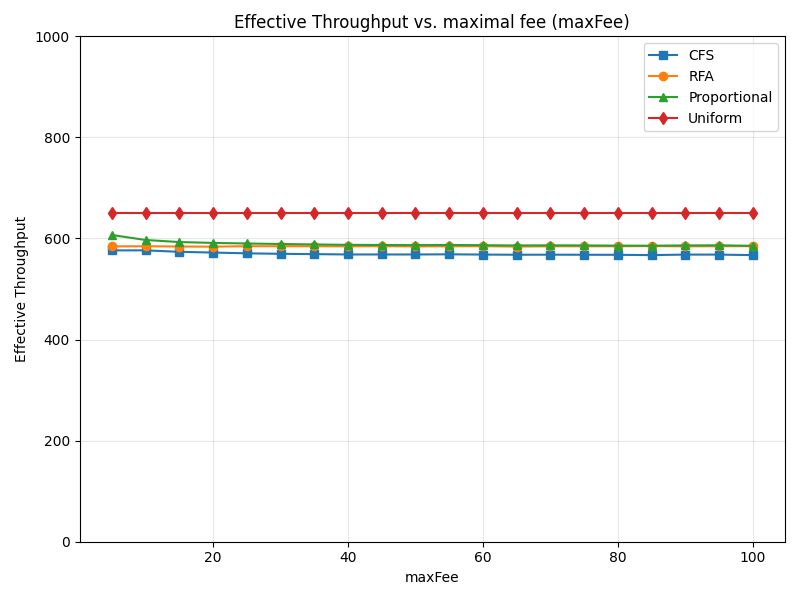}
    \caption{Effective Throughput vs. Maximum Fee (\(\text{maxFee}\)).}
    \label{fig:effective_throughput_maxFee}
\end{figure}

\subsection{Impact of Varying the Zipf Parameter (\(s\))}
\label{sec:results_s_variation}

To analyze the impact of skewed fee distributions on the performance of different strategies, we vary the Zipf parameter (s) from 0 (uniform distribution) to 1.4 in increments of 0.1, while keeping all other parameters fixed. Figures~\ref{fig:fee_throughput_s} and~\ref{fig:effective_throughput_s} present the expected fee throughput and effective throughput results, respectively.

\paragraph{Expected Fee Throughput}
Figure~\ref{fig:fee_throughput_s} shows that CFS, closely followed by PTS, achieves the highest fee throughput across all values of s. However, as s increases, fee throughput declines for all strategies. This is because higher values of s concentrate fees on fewer transactions, increasing the likelihood of collisions. The RFA model again performs slightly worse than RTS, with the difference diminishing as heterogeneity increases.

\paragraph{Effective Throughput}
Figure~\ref{fig:effective_throughput_s} presents the results for effective throughput. The uniform strategy (RTS) maintains the highest effective throughput, while RFA appears robust to changes in s. In contrast, Proportional and CFS exhibit lower effective throughput as s increases, as they prioritize high-fee transactions at the expense of covering a broader range of transactions.

\begin{figure}[!h]
    \centering
    \includegraphics[width=0.8\linewidth]{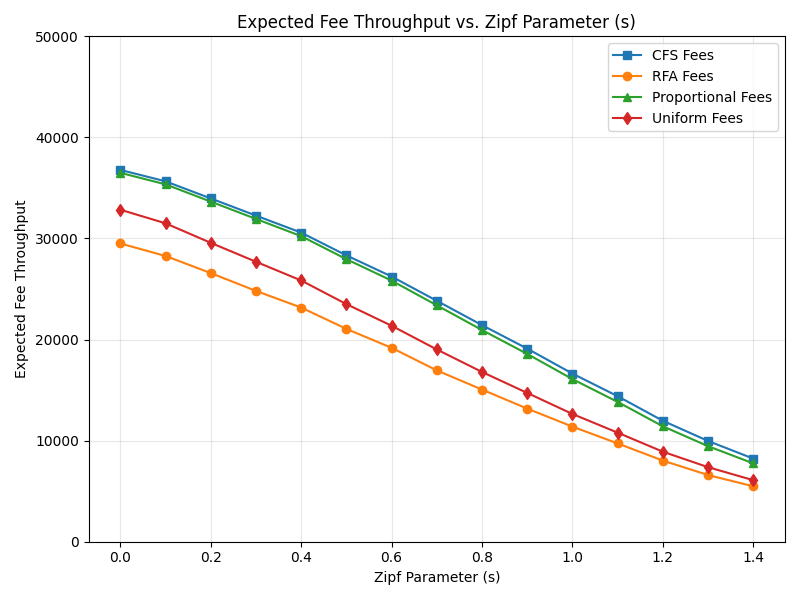}
    \caption{Expected Fee Throughput vs. Zipf Parameter (\(s\)).}
    \label{fig:fee_throughput_s}
\end{figure}

\begin{figure}[!h]
    \centering
    \includegraphics[width=0.8\linewidth]{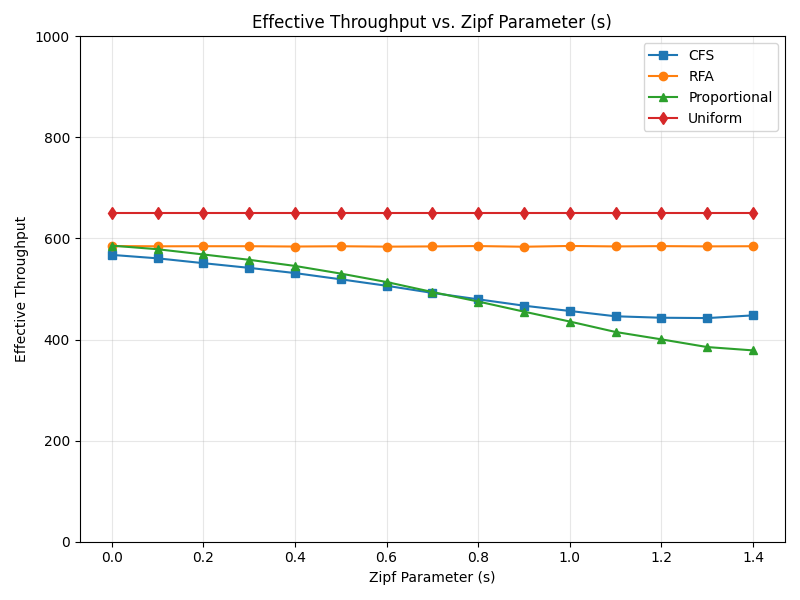}
    \caption{Effective Throughput vs. Zipf Parameter (\(s\)).}
    \label{fig:effective_throughput_s}
\end{figure}

\subsection{Discussion.}
Across all three scenarios, the CFS model consistently achieves the highest fee throughput, making it well-suited for applications where maximizing fee revenue is the primary objective. However, this comes at the expense of lower effective throughput, particularly in skewed fee distributions where high-fee transactions dominate selection at the cost of broader inclusion.

In contrast, the RFA model is consistently outperformed, often yielding lower fee throughput than even uniform selection (RTS). This highlights that its Nash equilibrium emerges naturally from the incentives and the freedom given to players in choosing their strategy, rather than guaranteeing an optimal system-wide outcome. While RFA allows for decentralized decision-making, it can lead to inefficiencies that negatively impact both individual players and overall network performance, demonstrating that economic freedom in protocol design does not always align individual incentives with collective welfare.


\section{Related Work}
Our work builds on the foundational results of \cite{inclusive}, which derived symmetric Nash equilibria (NE) for a single-shot transaction selection game under a specific fee-sharing model. While this work established the theoretical existence of equilibria, it did not calculate them explicitly, leaving room for further exploration. In contrast, we provide efficient methods to compute these equilibria for both Random Fee Allocation (RFA) and Collaborative Fee Sharing (CFS), enabling a detailed performance comparison.

Effective throughput in DAG-based distributed ledger technologies (DLTs) has been previously studied in sequential scenarios, where transaction selection strategies play a key role in maximizing network efficiency. For example, \cite{Peresni2021DAGOrientedPP} analyzed transaction selection strategies in protocols like PHANTOM and GHOSTDAG \cite{Sompolinsky2018PHANTOMG}, demonstrating that deviations from honest transaction selection strategies by malicious actors can harm network throughput and benefit attackers. These works emphasize the impact of sequential decision-making but do not focus on the single-shot game model analyzed in our study.

In \cite{IncentiveAttackswithRTS}, the authors investigate Random Transaction Selection (RTS) strategies used in various DAG-based protocols, including PHANTOM and GHOSTDAG \cite{Sompolinsky2018PHANTOMG}, and SPECTRE~\cite{Sompolinsky2016SPECTREAF}. They show that RTS does not constitute an NE and that greedy deviations from this strategy can harm throughput while increasing individual rewards. Their analysis focuses on sequential game-theoretic scenarios and the possibilities of having mining pools, contrasting with our single-shot NE analysis, which models transaction selection at a single block creation round.

The current Sui network employs random latency delays to control transaction selection, extending the model beyond a single-shot game. However, proposals such as \cite{sui-sips-45} suggest removing these delays to align transaction inclusion probability with gas price. This adjustment would make the single-shot game model more applicable, as validators prioritize immediate transaction inclusion based on gas price.

 Game theory is a classic topic in the study of the dynamics of DLTs. The most relevant previous works are \cite{altman2020blockchain, altman2019mining}. Here, the authors study competition among miners, deciding how much computational resources to allocate to solving a cryptographic puzzle. This work shares similarities with ours as they study a resource allocation game. Moreover, they also show that certain Nash equilibria (NE) can be found by solving a concave optimization problem. However, there are two main differences between these works and ours. First, they consider computational power as the resource and focus on pure NE, which leads more toward an ``economic analysis''. Their work does not lead to an implementable strategy for validators. Second, they concentrate on controlling the continuous rate at which a puzzle is solved rather than investigating the optimal fee allocation mechanism. 

In summary, while previous works have examined throughput and incentive effects in DAG-based protocols, they have primarily considered sequential or heuristic strategies. To our knowledge, explicit computation of NE for single-shot transaction selection under RFA and CFS models is addressed here for the first time.

\section{Discussion and Future Work}

This study examined transaction selection under a single-shot model, focusing on two fee allocation mechanisms and their Nash equilibria. While the results provide insight into the behaviour of validators in this simplified setting, several limitations and opportunities for further research remain.

An important direction is to extend the analysis to multi-round or repeated settings, where transaction selection occurs iteratively. Decisions made in one round can influence future availability and strategies, making sequential models more reflective of real-world behaviour. Such extensions could also explore how strategies adapt dynamically and whether they converge over time.

The static fees and homogeneous latencies assumptions simplify the model but do not fully capture real-world conditions. Validators face varying network delays, and transaction fees fluctuate over time. Incorporating dynamic fees and heterogeneous latencies into the analysis could offer a more realistic understanding of the incentives and challenges in transaction selection.

Another important consideration is the impact of block DAG structures, where block positions and dependencies influence the rewards for including specific transactions. Expanding the model to account for these dependencies could give a complete picture of how validators may optimize their strategies in sequential protocols.

The role of validator cooperation is particularly relevant in the context of Collaborative Fee Sharing (CFS), which inherently encourages joint strategies. Mechanisms enabling pooling or coordination among validators through mining pools or direct communication could significantly increase throughput by reducing transaction duplication. 

Finally, while this work provides a method to compute Nash equilibria, the computational complexity of doing so in real time, especially for a large number of transactions, remains a challenge. Developing efficient heuristics that approximate the Nash equilibrium could make these strategies practical for implementation. Evaluating the performance of these heuristics in dynamic and adversarial settings would also be a valuable next step.

\section*{Disclaimer on the Use of AI Assistance}

This paper partially leverages AI assistance, specifically through GPT-based tools and Grammarly, to enhance the writing, editing, and drafting process. We utilized these tools in the following ways:

While these tools accelerated specific aspects of the workflow, the authors independently developed and critically reviewed the theoretical derivations, numerical methods, and final scientific conclusions presented in this paper.

Additionally, we attempted to use GPT-4 and GPT-o1 to directly solve the optimization problems of calculating Nash equilibria (NE) for the models under study. Despite extensive experimentation, these attempts failed to identify a practical or computationally efficient solution for finding the NE. The authors independently designed and implemented the eventual solution methodology based on concave optimization techniques and numerical methods without reliance on AI tools.

These tools were used only for auxiliary tasks and language refinements. They were not involved in any substantive or original aspect of this paper's scientific contributions.

\section{Appendix}
The following identity provides an interesting exercise for an undergraduate probability course. We include it here for the sake of completeness.
\begin{lemma}\label{lem:identity}
Let $p\in(0,1]$ and $N\in\mathbb{N}$, then 
    \begin{align}\label{eq:shareRFAappendix}
    f(p) &= \sum_{i=0}^{N-1} \frac{1}{i+1} {N-1 \choose i}  p^i (1-p)^{N-1-i} \cr
    & = \frac{1 - (1-p)^{N}}{N\, p}.
\end{align}
\end{lemma}

\begin{proof}
Start by rewriting \(\frac{1}{i+1}\) using the integral representation:
\[
\frac{1}{i+1} = \int_0^1 x^i \, dx.
\]
Substituting this into the summation, we have:
\[
f(p) = \int_0^1 \sum_{i=0}^{N-1} \binom{N-1}{i} (p x)^i (1-p)^{N-1-i} \, dx.
\]
The inner summation is the binomial expansion of \(((p x) + (1-p))^{N-1}\):
\[
\sum_{i=0}^{N-1} \binom{N-1}{i} (p x)^i (1-p)^{N-1-i} = (1 - p + p x)^{N-1}.
\]
Substituting back, we get:
\[
f(p) = \int_0^1 (1 - p + p x)^{N-1} \, dx.
\]

Now, make the substitution \(u = 1 - p + p x\), so \(du = p \, dx\), with limits:
\[
x = 0 \implies u = 1-p, \quad x = 1 \implies u = 1.
\]
The integral becomes:
\[
f(p) = \frac{1}{p} \int_{1-p}^1 u^{N-1} \, du.
\]

Evaluate the integral:
\[
\int_{1-p}^1 u^{N-1} \, du = \frac{u^N}{N} \Big|_{1-p}^1 = \frac{1^N - (1-p)^N}{N}.
\]
Thus:
\[
f(p) = \frac{1}{N p} \left( 1 - (1-p)^N \right).
\]
\end{proof}

\bigskip
\bibliographystyle{plain}
\bibliography{ref}

\end{document}